\documentclass[journal]{IEEEtran}

\usepackage{microtype}
\usepackage{graphicx}
\usepackage{booktabs} 
\usepackage{tipa}
\usepackage{upgreek}
\usepackage{subfig}
\usepackage{lipsum}

\usepackage{mathalpha}
\usepackage{hyperref}

\usepackage{amsmath,amssymb}
\usepackage{dsfont}
\usepackage{enumerate}
\usepackage{algorithm}
\usepackage{graphicx}      
 \usepackage[square,sort,comma,numbers]{natbib}        
\usepackage{comment}
\usepackage{algorithmic}
\usepackage{amsthm}
\newtheorem{remark}{Remark}
\newtheorem{theorem}{Theorem}
\newtheorem{lemma}[theorem]{Lemma}

\DeclareMathAlphabet{\mathpzc}{OT1}{pzc}{m}{it}
\DeclareMathOperator*{\argmin}{arg\,min}

\begin{document}

\title{Chance Constrained Policy Optimization for Process Control and Optimization}

 \author{Panagiotis Petsagkourakis, Ilya Orson Sandoval, Eric Bradford, Federico Galvanin, Dongda Zhang and Ehecatl Antonio del Rio-Chanona
 
 \thanks{P. Petsagkourakis and F. Galvanin are with the Centre for Process Systems Engineering (CPSE), University College London, Torrington Place, London, United Kingdom}
 \thanks{ I. O. Sandoval and Ehecatl Antonio del Rio-Chanona are with the Centre for Process Systems Engineering (CPSE), Imperial College London, UK}
 \thanks{ E. Bradford is with the Department of Engineering Cybernetics, Norwegian University of Science and Technology, Trondheim, Norway}
 \thanks{D. Zhang is with the Centre for Process Integration, Department of Chemical Engineering and Analytical Science,The University of Manchester, UK and Centre for Process Systems Engineering (CPSE), Imperial College London, UK}
 
 }

\maketitle

\begin{abstract}                
Chemical process optimization and control are affected by 1) plant-model mismatch, 2) process disturbances, and 3) constraints for safe operation. Reinforcement learning by policy optimization would be a natural way to solve this due to its ability to address stochasticity, plant-model mismatch, and directly account for the effect of future uncertainty and its feedback in a proper closed-loop manner; all without the need of an inner optimization loop. One of the main reasons why reinforcement learning has not been considered for industrial processes (or almost any engineering application) is that it lacks a framework to deal with safety critical constraints.
Present algorithms for policy optimization use difficult-to-tune penalty parameters, fail to reliably satisfy state constraints or present guarantees only in expectation. We propose a chance constrained policy optimization (CCPO) algorithm which guarantees the satisfaction of joint chance constraints with a high probability - which is crucial for safety critical tasks. This is achieved by the introduction of constraint tightening (backoffs), which are computed simultaneously with the feedback policy. Backoffs are adjusted with Bayesian optimization using the empirical cumulative distribution function of the probabilistic constraints, and are therefore self-tuned. This results in a general methodology that can be imbued into present policy optimization algorithms to enable them to satisfy joint chance constraints with high probability. We present case studies that analyze the performance of the proposed approach.

\end{abstract}
\section{Introduction}

The optimization of chemical processes presents distinctive challenges to the stochastic systems community given that they suffer from three conditions:
1) there is no precise known model for most industrial scale processes (plant-model mismatch), leading to inaccurate predictions and convergence to suboptimal solutions, 2) the process is affected by disturbances (i.e. it is stochastic), and 3) state constraints must be satisfied due to operational and safety concerns, therefore constraint violation can be detrimental or even dangerous. In this work we use constrained policy search, a reinforcement learning (RL) technique, to address the above challenges. 

RL is a machine learning technique that computes a policy which learns to perform a task by interacting with the (stochastic) environment.
RL has been shown to be a powerful control approach, and one of the few control techniques able to handle nonlinear stochastic optimal control problems \cite{Bertsekas:2000:DPO:517430, nathan2019}. RL in the approximate dynamic programming (ADP) sense has been studied for chemical process control in \cite{Lee2005}. A model-based strategy and a model-free strategy for control of nonlinear processes were proposed in \cite{Peroni2005}. ADP strategies were used to address fed-batch reactor optimization, in \cite{Lee2006} mixed-integer decision problems were addressed with applications to scheduling. In \cite{Tang2018} RL was combined with distributed optimization techniques, to solve an input-constrained optimal control, among other works (e.g. \cite{Chaffart2018}, \cite{Shah2016}). All these approaches rely on action-value methods, which approximate the solution of the Hamilton–Jacobi–Bellman equation (HJBE), and have been shown to be reliable for some problem instances. RL also shares similar features with multiparametric model predictive control (or explicit MPC) \cite{Bemporad1999, CharitopoulosEMPC}, where an explicit representation of the controller can be found, such that it satisfies the optimality conditions. However, explicit MPC usually requires some approximations or assumptions on the models.

Policy gradient RL methods \cite{Sutton:1999:PGM:3009657.3009806} have been proposed to directly optimize the control policy. Unlike action-value RL methods where convergence to local optima is not guaranteed, policy gradient methods can guarantee convergence to local optimality even in high dimensional continuous state and action spaces. 
However, the inclusion of constraints in policy gradient methods is not straightforward. Existing methods cannot guarantee strict feasibility of the policies even when initialized with feasible initial policies \cite{Wen2018}. The main approaches to incorporate constraints make use of trust-region, fixed penalties~\cite{Achiam2017, Tessler2018}, and cross entropy~\cite{Wen2018}. Furthermore, if online optimization is to be avoided, addressing the constraints by the use of penalties is a natural choice. Various approaches have been proposed in this direction, however, current approaches easily lose optimality or feasibility ~\cite{Achiam2017} and  guarantee feasibility only in expectation. Following this thread of thought, a Lyapunov-based approach is implemented in \cite{Chow2019}, where a Lyapunov function is constructed and the unconstrained policy is projected to a safety layer allowing the satisfaction of constraints in expectation. In \cite{Yang2020Projection-Based} an upper bound on the expected constraint violation is provided for the projected-based policy optimization, where the initial unconstrained policy is projected back to the constraint set. In \cite{Liu2019} an interior-point inspired method that is widely used in control \cite{PetsagkourakisTAC} is proposed, where constraints are incorporated into the reward using a logarithm barrier function, allowing the satisfaction of the constraints in expectation. 

The above methods all guarantee constraint satisfaction in expectation, which is inadequate for safety critical engineering applications (very loosely speaking, this means violating 50\% of the time). Furthermore, although penalties are a natural way to avoid an online optimization loop (one of the main advantages of policy gradients), tuning these penalties is not always straightforward, and usually rely on heuristics which significantly affect the performance of policy gradient algorithms \cite{Engstrom2020Implementation}. This suggests constrained policy gradient methods would benefit from self-tuning parameters.
As mentioned earlier, it is well known that policy gradient methods present many advantages, however, their application domain will remain limited until they can handle constraint satisfaction reliably. This is the main challenge addressed in this work.
Our proposed method - chance constrained policy optimization (CCPO) - guarantees the satisfaction of joint chance constraints for the optimal policy. This allows for the satisfaction of constraints with a high probability, rather than only in expectation. To achieve the satisfaction of joint chance constraints without the need of an online optimization, we use tightening of constraints, which are appended to the objective. We treat the problem of finding the least tightening that satisfies the probabilistic constraints as a black-box optimization problem, and can therefore solve it efficiently by existing methods. Upon convergence the proposed algorithm finds an optimal policy which guarantees the constraint satisfaction to the desired tolerance.

The structure of this paper is as follows. The problem statement is outlined in section 2, the details of the proposed method for probabilistic satisfaction in RL is presented in section 3. A case study is presented in section 4, where the framework is applied to a dynamic bioreactor system, and in the last section, conclusions are outlined.
%
%
\section{Problem Statement}
In this work, the dynamic system is assumed to be given by a probability distribution, following a Markov process,
\begin{equation}\label{MC_unc}
    \textbf{x}_{t+1}\sim p(\textbf{x}_{t+1}|\textbf{x}_t, \textbf{u}_t),\quad\textbf{x}_0 \sim p(\textbf{x}_0),
\end{equation}
where $p(\textbf{x}_i)$ is the probability density function of $\textbf{x}_i$, with $\textbf{x}\in \mathbb{R}^{n_x}$ representing the states, $\textbf{u}\in \mathbb{R}^{n_u}$ the control inputs and $t$ discrete time. This behaviour is observed in systems when stochastic disturbances are present and/or other uncertainties affect the physical system, like parametric uncertainties. A discrete-time system with disturbances and parametric uncertainties can be written as:
\begin{equation}\label{additive_unc}
        \textbf{x}_{t+1} =f(\textbf{x}_{t},\textbf{u}_t, \textbf{p}, \textbf{w}_t)
\end{equation}
where $\textbf{w}\in \mathbb{R}^{n_w}$ is a vector of  disturbances and $\textbf{p}\in \mathbb{R}^{n_p}$ are uncertain parameters. This can be represented by (\ref{MC_unc}).
In this work we seek to maximize an objective function in expectation by using an optimal stochastic policy subject to probabilistic constraints despite the uncertainty of the system. This problem can be written as a stochastic optimal control problem (SOCP):
\begin{equation}
\mathcal{P}(\pi(\cdot)):=\left
\{\begin{aligned}
        &\max_{\pi(\cdot)} \mathbb{E}\left(J(\textbf{x}_{0},\dots,\textbf{x}_{T},\textbf{u}_{0},\dots,\textbf{u}_{T})\right)\\
    &\text{s.t.}\\
    &\textbf{x}_0 \sim p(\textbf{x}_0)\\
    &\textbf{x}_{t+1} \sim p(\textbf{x}_{t+1}|\textbf{x}_t, \textbf{u}_t)\\
    &\textbf{u}_t \sim p(\textbf{u}_t|\textbf{x}_t)=\pi(\textbf{x}_t)\\
    &\textbf{u}_t\in\mathbb{U}\\
    &\mathbb{P}(\bigcap_{i=0}^T \{\textbf{x}_i \in \mathbb{X}_i\})\geq 1-\alpha\\
    &  \forall t \in \left\{0,...,T-1\right\}\label{eq:OCP}
\end{aligned}\right.
\end{equation}
where $J$ is the objective function, $\mathbb{U}$ denotes the set of hard constraints for the control inputs, $\mathbb{X}_i$ represents constraints for states that must be satisfied with a probability $1-\alpha$. Specifically, 
\begin{equation}\label{constraint}
    \mathbb{X}_t = \{\textbf{x}_t\in \mathbb{R}^{n_x}|g_{j, t}(\textbf{x}_t) \leq 0, j =1,\dots,n_g\},
\end{equation} 
with $g_{j,t}$ being the $j^{th}$ constraint to be satisfied at time instant $t$ and the joint chance (also called probabilistic) constraints ($\mathbb{P}(\bigcap_{t=0}^T \{\textbf{x}_t \in \mathbb{X}_i\})\geq 1-\alpha$) are satisfied for the full trajectory over all $t\in\{0,\dots,T\}$.  The probability density function of $\textbf{u}_t$ given state $\textbf{x}_t$ is $p(\textbf{u}_t|\textbf{x}_t)$, and $\pi(\textbf{x}_t)$ is the state feedback policy.
Additionally, $\pi(\cdot)$ is the stochastic policy. Unfortunately, this SOCP is generally intractable and approximations must be sought, hence the use of RL \cite{Bertsekas:2000:DPO:517430, schulman2017proximal, Petsagkourakis2019a}.

In RL a policy $\pi_\theta(\cdot)$ parametrized by the parameters $\theta$ is constructed. This policy maximizes the expectation of the objective function $J(\cdot)$. In the finite horizon discrete-time case, this objective function can be defined as
\begin{equation}
    J = \sum_{t=0}^T\text{\textgamma}^{t}R_t(\textbf{u}_t,\textbf{x}_t),
\end{equation}

where $\text{\textgamma}\in [0,1]$ is the { discount factor} and $R_t$ a given reward at the time instance $t$ for the values of $\textbf{u}_t$, $\textbf{x}_t$.

Notice that we seek for the policy to satisfy joint chance constraints with some high probability. Previous approaches have address the satisfaction of constraints in expectation; this means that constraints are violated (roughly speaking) half of the time.

%
\section{Constrained Policy Optimization for Chance Constraints}
In RL agents take actions to maximize some expected reward given a performance metric. In process control, these agents become the controller, which use a feedback policy $\pi(\cdot)$, to optimize the expected value of an economic metric of the process ($J$). The physical system (or the model) at each sampling time produces a value for the reward $R$ which reflects the performance of the policy. The RL algorithm determines the feedback policy that produces the greatest reward in expectation, which is referred to as policy optimization. However, there is no natural way to handle constraints, and a constraint satisfied in expectation may still have a very high probability of not being satisfied. To satisfy the constraints with some high probability and not only in expectation, we tighten the constraint with backoffs~\cite{Bradford2019, Rafiei2018} $b_{j,t}$ as:
\begin{equation}\label{constraint1}
    \mathbb{\overline{X}}_t = \{\textbf{x}_t\in \mathbb{R}^{n_x}|g_{j, t}(\textbf{x}_t)+b_{j,t} \leq 0, j =1,\dots,n_g\},
\end{equation} 
where variables $b_{j,t} \geq 0$ represent the backoffs which tighten the original feasible set $\mathbb{X}_t$ defined in (\ref{constraint}). Backoffs restrict the perceived feasible space by the controller, and allow guarantees on the satisfaction of chance constraints.

We denote  $\pmb{\tau}$, as the joint random variable of states, controls and rewards for a trajectory with a time horizon $T$:   
\begin{equation}\label{tau}
    \pmb{\tau} = (\textbf{x}_0,\textbf{u}_0,R_0,...,\textbf{x}_{T-1},\textbf{u}_{T-1},R_{T-1},\textbf{x}_T,R_T).
\end{equation}
We also assume that the policy can be parametrized by a finite number of parameters $\theta$, and denote this parametrized policy as:
\begin{equation}\label{policyopt}
    \textbf{u}_t \sim p(\textbf{u}_t|\textbf{x}_t)= \pi_\theta(\textbf{x}_t,D_t),\quad \textbf{u}_t \in \mathbb{U},
\end{equation}
where $\textbf{u}_t \in \mathbb{U}$ is inherently satisfied by the construction of the policy, e.g. the policy passes through a bounded and differential squashing function~\cite{Deisenroth2015}, and $D_t$ is a window of past inputs and states that are used by the policy, for example, if the parametrized policy is a recurrent neural network, then $D_t$ corresponds to a number of past states and controls.

Recently a methodology was proposed that satisfies the expected value of the constraints~\cite{Tessler2018, Wen2018}, however this is not adequate for safety critical constraints in chemical processes. Instead, to account for constraint violations, we incorporate probabilistic constraints. The problem can then be reformulated as:
\begin{equation}\label{policyopt_con}
\begin{split}
    \pi_{\theta^*} =& \arg\max_{\pi_\theta(\cdot)} \mathbb{E}_{\pmb{\tau}}\left(J(\pmb{\tau})\right)\\
    &s.t.\\
    &\textbf{u}_t \sim p(\textbf{u}_t|\textbf{x}_t)= \pi_\theta(\textbf{x}_t,D_t), ~\forall t\in {0,...,T-1} \\
    &\textbf{u}_t \in \mathbb{U} \\
    & \pmb{\tau}\sim p(\pmb{\tau}|\theta) \\
    & \mathbb{P}_{\pmb{\tau}}\left(\bigcap_{i=1}^T \{\textbf{x}_i \in \mathbb{X}_i\}\right)\geq 1-\alpha
    \end{split}
\end{equation}
where $p(\pmb{\tau}|\theta)$ represents the probability of the trajectory $\pmb{\tau}$ given the parametrization $\theta$ of the feedback policy, and $\pi_{\theta^*}$ is the optimal policy.

In order to solve~(\ref{policyopt_con}) we propose to: 
\begin{enumerate}
    \item Parameterize the feedback policy by a multilayer neural network that computes the mean and variance of the control actions, which are consequently sampled as a normal distribution resulting in a stochastic policy.
    \item The probabilistic constraint in (\ref{policyopt_con}) is substituted by a tightened constraint set $\mathbb{\overline{X}}_t$ to guarantee closed-loop probabilistic constraint satisfaction.
    \item The tightened constraints $g_{j, t}(\textbf{x}_t)+b_{j,t} \leq 0, j =1,\dots,n_g$ are incorporated into the objective function and can be incorporated in previous approached \cite{Achiam2017, Chow2018}. This avoids the need to explicitly solve an optimization problem at run-time.
    \item The policy optimization is performed by a policy gradient framework~\cite{Sutton:1999:PGM:3009657.3009806}.
    \item The value of the backoffs should be the smallest value that guarantees constraint satisfaction with a probability of at least $1-\alpha$. This is formulated as a black-box optimization problem, which can be solved efficiently by existing methods \cite{Zhan2020, Jones1998, Cartis2018}.
\end{enumerate}

Notice that the value of the backoffs imply a trade-off: large values guarantee constraint satisfaction, but they make the problem over-conservative and mitigate performance, while smaller values produce solutions with high rewards, but may not guarantee the constraint satisfaction to the desired accuracy ($1-\alpha$). Also notice that if backoffs are too large, the problem might become infeasible.

In the next subsections we introduce the components that were outlined above.


\subsection{Policy parametrization}
For the policy parametrization we use recurrent neural networks, (RNNs)~\cite{Rumelhart1986, RNN2017}. Let $N$ be the length of the window of past states and controls to be used by the parametrized feedback policy, i.e. $D_t = \left[\textbf{x}_{t-1}^T,\textbf{u}^T_{t-1},\textbf{x}_{t-2}^T,\textbf{u}_{t-2}^T,\dots, \textbf{u}_{t-N-1}^T \right]^T$. Then the stochastic policy can be defined as:
\begin{equation}\label{rnn}
     \pi_\theta(\textbf{x}_{t},D_t) = \mathcal{N}(\textbf{u}_t|\pmb{\mu}^\textbf{u}_t,\pmb{\Sigma}^\textbf{u}_t), \quad [\pmb{\mu}^\textbf{u}_t,\pmb{\Sigma}^\textbf{u}_t]=s_\theta(\textbf{x}_{t},D_t)
\end{equation}
where $s_\theta$ is the multilayer RNN parametrized by $\theta$, and $\pmb{\mu}^\textbf{u}_t$ and $\pmb{\Sigma}^\textbf{u}_t$ are the mean and covariance of the normal distribution from which the control $\textbf{u}_t$ is drawn. Deep structures are employed to enhance the performance of the learning process ~\cite{Mnih2013,Mnih2015}.


\subsection{Probabilistic constraints}
In this section, we present a strategy to handle  (joint) chance constraints by policy gradient methods. There is no closed form solution for general probabilistic constraints, and therefore we compute their empirical cumulative distribution function (ECDF) instead. To guarantee with high probability (and not only in expectation) the satisfaction of constraints, we introduce constraint tightening by using backoffs $b_{j,t}$ (see Eq. (\ref{constraint1})). Therefore, we conduct closed-loop Monte Carlo (MC) simulations to compute the backoffs. We pose the problem of attaining the smallest backoffs (least restrictive) that still satisfy our constraints with a probability of at least $1-\alpha$ as an expensive black-box optimization problem, which is solved by Bayesian optimization. In this way we compute a policy that satisfies our constraints with a probability of at least $1-\alpha$, and we can be certain of this with a confidence of at least $1-\epsilon$. The rest of this subsection details the aforementioned methodology.

For convenience we define a single-variate random variable $C(\cdot)$ representing the satisfaction of the joint chance constraint:
\begin{align}\label{eq:chanceconstr_full}
    & C(\mathbf{X}) = \max_{(j,t) \in \{1,\ldots,n_g\} \times \{1,\ldots,T\}} {g_{j,t}(\mathbf{x}_t)}\\
    & F(c) = \mathbb{P}\left(C(\mathbf{X}) \leq c\right) 
\end{align}
then 
\begin{align}\label{eq:chanceconstr}
    & F(c:=0)=  \mathbb{P}\left(C(\mathbf{X}) \leq 0\right) = \mathbb{P} \left(\bigcap^T_{t=0} \{ \mathbf{x}_t \in \mathbb{X}_t \}    \right), 
\end{align}
where $\mathbf{X} = [\mathbf{x}_1,\ldots,\mathbf{x}_T]^{\sf T}$, and $F$ is the cumulative distribution function (CDF).
There is no analytical expression for general probabilistic constraints, and we therefore approximate them by a non-parametric approximation, i.e. their empirical cumulative distribution function (ECDF). We can approximate the value of $F(0)$ by its sample approximation on $S$ Monte Carlo (MC) simulations of $\mathbf{X}$:
\begin{equation}\label{ECDF}
F(0) \approx F_S(0) = \dfrac{1}{S}\sum_{s=1}^S \mathds{1}\left(C(\mathbf{X}^s) \leq 0 \right),
\end{equation}
where $\mathbf{X}^s = [\mathbf{x}^s_1,\ldots,\mathbf{x}^s_T]^{\sf T}$ is the $s^\text{th}$ MC sample of the state trajectory. Notice that ${F}_S(0)$ is a random variable. Define  $\mathds{1}\{C(\mathbf{X}) \leq 0 \}$ as the indicator function for a single trajectory to satisfy all constraints:
%
\begin{equation*}
\mathds{1}\{C(\mathbf{X}) \leq 0\ \} = \begin{cases}1,&  C(\mathbf{X}) \leq 0
\\ 0,& otherwise\end{cases}.
\end{equation*}
Notice that $F_S(0)$ follows a Binomial distribution, as the indicator function is a Bernoulli random variable.
Hence, ${F}_S(0) \sim \dfrac{1}{S}\text{Bin}(S, F(0))$. Later on we will use a realization from this random variable (i.e. $\hat{F}_S$) to approximate the probability for a trajectory to satisfy all constraints. The confidence bound for the ECDF can then be computed from the Binomial CDF via the Clopper-Pearson interval ~\cite{10.2307/2331986, brown2001interval, Bradford2019}. 

Notice that the quality of the approximation in (\ref{ECDF}) strongly depends on the number of samples $S$ used and it is therefore desirable to quantify the uncertainty of the sample approximation itself. This problem has been studied to a great extent in the statistics literature \cite{10.2307/2331986, brown2001interval}, leading to the following lemma. 

\begin{lemma}[\cite{brown2001interval}]\label{lower}
Given the random variable $F_S(0)$ (the ECDF - see (\ref{ECDF}))  based on $S$ i.i.d. samples, the true value of the CDF, $F(0)$, has the following lower bound ${F}_{lb}$ with a confidence level of $1-\epsilon$:
\begin{equation}\label{eq lower bound}
    \begin{split}
        &\mathbb{P}(F(0) \geq F_{lb})\geq 1-\epsilon,\\
        &F_{lb} = 1-\textup{betainv}({\epsilon}, S + 1-S~F_S(0), S~F_S(0)),
    \end{split} 
\end{equation}
with $\textup{betainv}(\cdot,\cdot,\cdot)$ being the inverse of the beta CDF with parameters  $\{S + 1-S~F_S(0)\}$ and $\{S~F_S(0)\}$.
\end{lemma}
Lemma \ref{lower} states that it is possible to compute a probabilistic lower bound, with an arbitrarily confidence of at least $1-\epsilon$, for the ECDF of the chance constraints. The following lemma follows directly using the above result.

\begin{lemma}\label{satisfaction}
Given a realization $\hat{F}_{S}$ of the ECDF random variable $F_{S}(0)$ based on $S$ independent samples, we can compute a corresponding realization of the random variable $F_{lb}$ which is denoted $\hat{F}_{lb}$. If the value of $\hat{F}_{lb} \geq 1-\alpha$, then, with a confidence level of $1-\epsilon$, our original chance constraint in \textup{Eq. (\ref{eq:chanceconstr})}  holds.  
\end{lemma}
%
%

The implication of Lemma \ref{satisfaction} is that we can  guarantee the satisfaction of joint chance constraints with a user-defined probability of at least $1-\alpha$ and a user-defined confidence of at least $1-\epsilon$.
%
Values of $\hat{F}_{lb}$ that are higher than $1-\alpha$ lead to more conservative solutions, and consequently worse values for the objective function. On the other hand, when $\hat{F}_{lb}$ is lower than $1-\alpha$, then the objective attains better values, but the constraints are not satisfied to the required degree. Therefore, the best trade-off is realized  if $\hat{F}_{lb}$ is as close as possible to $1-\alpha$. 

Hence, we aim to compute tightened constraints employing backoffs $b_{j,t}$ such that $\hat{F}_{lb}-(1-\alpha)$ is close to 0 when the policy optimization has terminated.  Lemma \ref{satisfaction} allows us to find a solution to the original problem having satisfied all joint chance constraints with a confidence at least $1-\epsilon$.

To compute the tightened constraint set as denoted in Eq. (\ref{constraint1}), we first compute an initial set of backoffs ($b_{j,t}^0$), as it has been proposed in~\cite{Paulson2018}, where 
\begin{equation}
    \mathbb{E}\left(g_{j,t}(\textbf{x}_t)\right) + b^0_{j,t} =0~~\text{gives}~ \mathbb{P}\left(g_{j,t}(\textbf{x}_t)\leq 0\right)\geq 1-\delta~~ \forall j,t, 
\end{equation}
with $\delta$ being a tuning parameter for the initial backoffs $b^0_{j,t}$. 
Additionally, the expected value of $g_{j,t}$ is approximated with a sample average approximation (SAA):
\begin{equation}
\bar{g}_{j,t} =\frac{1}{S}\sum_{s=1}^Sg_{j,t}(\mathbf{x}_t^s).
\end{equation}

The initial backoff values are computed to probabilistically satisfy each constraint:
\begin{subequations}
\begin{align}
 &\mathbb{P}\left(g_{j,t}(\textbf{x}_t)\leq 0\right)\geq 1-\delta, \label{prob_in}\\
    &b^0_{j,t} = F_{\mathbf{g}_{j,t}}^{-1}(1-\delta) - \bar{g}_{j,t}(\textbf{x}_t),~\forall j,t
\end{align}
\end{subequations}
with $\mathbf{g}_{j,t}=[g_{j,t}(\mathbf{x}_t^1),\ldots,g_{j,t}(\mathbf{x}_t^S)]^{\sf T}$ and $F_{\mathbf{g}_{j,t}}^{-1}(1-\delta)$ is the $1-\delta$ quantile of $\mathbf{g}_{j,t}$ (\ref{prob_in}). 

We wish the backoffs $\textbf{b}$ (with elements $b_{j,t}~\forall~j,t~\in \{1,\dots,n_g\}\times\{1,\dots,T\}$) to be large enough to ensure the probabilistic satisfaction of constraints. However, if the backoffs are too large it may result in a conservative solution, therefore a worse performance, or even infeasibility of the problem. To obtain the least conservative solution that still guarantees the probabilistic constraint satisfaction specified by $1-\alpha$ we solve a root-finding problem using Lemma~\ref{lower} to find a vector $\pmb{\gamma}$ that parametrizes $b_{j,t} := \gamma_{j}~b^0_{j,t}~\forall j, t$ such that
\begin{equation}\label{solve_backoff}
   \hat{F}_{lb}(\pmb{\gamma})-(1-\alpha)= 0
\end{equation}
Notice, that now $\hat{F}_{lb}$ is a function of $\pmb{\gamma}=[\gamma_{1},...,\gamma_{n_g}]$. In  fact the Eq. \ref{solve_backoff} is an ideal scenario, where the lower bound could be forced to be $1-\alpha$. Here we try to minimize the distance between $\hat{F}_{lb}$  and $1-\alpha$.
With the above procedure we compute a deterministic surrogate for the  constraints that allows us to satisfy the joint chance constraints in the original optimization problem (\ref{policyopt_con}). In the subsequent section we explain how this constraint surrogate is incorporated into the reinforcement learning framework.

\section{CCPO: Chance Constrained Policy Optimization}\label{alg}

\subsection{Policy gradient for fixed backoffs}

In this work we reformulate chance constraints such that their satisfaction is guaranteed with high probability, and they can be incorporated into the objective of the policy gradient method \cite{Sutton:1999:PGM:3009657.3009806}. We therefore avoid the need for a numerical optimization every time the agent/controller outputs an action/control input.

Loosely speaking policy gradient methods aim to update the parametrized policy using the gradient of the reward. Without loss of generality, we outline the implementation for the REINFORCE~\cite{williams1992simple} algorithm for ease of presentation, but this can be generalized to any policy gradient or actor-critic method \cite{schulman2017proximal, Kakade2002}. 

Previous works have embedded Lagrangian methods as adaptive penalty coefficients to enforce satisfaction of constraints \cite{Tessler2018}, however, an adaptive scheme such as gradient ascent or its variants on inequality multipliers are difficult to justify in theory (see for example \cite{Nocedal2006} chapter 12 or \cite{PrimalDualWright} chapter 5), and in practice tend to have numerical issue when an arbitrary number of constraints are enforced and different constraints are active in different instances \cite{Nocedal2006}. 

In this work, we instead propose a $p$-norm :
\begin{equation} \label{eq:modified_obj}
\hat{J}(\pmb{\tau}, \textbf{b}) = J(\pmb{\tau}) - \kappa~\sum_{t=1}^T{||\left[\textbf{g}_{t} (\textbf{x}_t) + \textbf{b}_{t}\right]^{-}}||_p^{p}, 
\end{equation}
where $\left[{g}_{j,t} (\textbf{x}_t) + b_{j,t}\right]^{-} =\max\{{g}_{j,t}, (\textbf{x}_t) + b_{j,t}~, 0\}$, $||\cdot||_p$ is a $p$-norm of the vector $\textbf{g}_t = \left[g_{1,t},\dots, g_{n_g,t}\right]$. We advocate for the use of $l_1$~($p=1$) and $l_2$~($p=2$) norms due to the arbitrary number of constraints that may be considered and their numerical stability~\cite{Nocedal2006}.

%
The following theorem can be stated for the satisfaction of the constraints given $S$ Monte Carlo trajectories. 
\begin{theorem}\label{thm_feasibility}
Consider the chance constrained stochastic optimal problem~(\ref{policyopt_con}) and let  $\pi_{\theta^*}$ be the trained policy that maximizes the expected reward (see \cite{Achiam2017,Tessler2018} or (\ref{eq:modified_obj})) and satisfies (\ref{solve_backoff}) using backoffs $\textbf{b}$. Then the joint chance constraints (\ref{eq:OCP}) will be satisfied with a confidence level of $1-\epsilon$. 
\end{theorem}
\begin{proof}
Assume that a policy $\pi_{\theta^*}$ is evaluated on the chance constraints (\ref{constraint}) and a realization of the ECDF $\hat{F}_S$ is computed which leads to a realization $\hat{F}_{lb}$ of the random variable ${F}_{lb}$. Then, given Lemma 2, if $\hat{F}_{lb}\geq1-\alpha$ policy $\pi_{\theta^*}$ will satisfy (\ref{eq:chanceconstr}) with confidence $1-\epsilon$.

\end{proof}
\begin{remark}
Further a posteriori analysis can be performed in the system following the work of \cite{Campi2018}. Given the number of trajectories $S$ and a confidence level $1-\epsilon_S$, then the probability of constraint violation can be found to be bounded with a confidence $1-\epsilon_S$. More details can be in \cite{Campi2018}
\end{remark}
\begin{remark}
The feasibility of the constraints will guaranteed independently from the selection of $\kappa$ (see Theorem~\ref{thm_feasibility}).
\end{remark}
\begin{remark}
The parameter $\kappa$ can be updated using standard techniques from constrained optimization~\cite{Nocedal2006}.
\end{remark}

We now use the policy gradient theorem \cite{Sutton:1999:PGM:3009657.3009806} to obtain an explicit gradient estimate of the reward with respect to the parameters of our policy:
\begin{align}
\nabla_\theta \mathbb{E}_{\pmb{\tau}} \left(\hat{J}\right) \approx \frac{1}{S} \sum_{s=1}^{S}  \left[ \hat{J}(\pmb{\tau}^{s},\textbf{b}) \nabla_\theta \sum_{t=0}^{T-1} \text{log}\left( \pi_\theta(\textbf{x}_t^s,D^s_t)\right)\right]
\end{align}
where $\pmb{\tau}^{s}=(\textbf{x}_0^s,\textbf{u}_0^s,R_0^s,...,\textbf{x}_{T-1}^s,\textbf{u}_{T-1}^s,R_{T-1}^s,\textbf{x}_T^s,R_T^s)$ denotes the realization of the $s^{th}$ trajectory with $R^s$ being the reward for sample $s$ and
\begin{equation}
D^s_t = \left[(\textbf{x}_{t-1}^s)^T,(\textbf{u}_{t-1}^s)^T,
\dots, (\textbf{u}_{t-N-1}^s)^T \right]^T
\end{equation}
denotes the past states and controls used by the policy for sample $s$. 
The variance of the gradient estimate can be reduced with the aid of an action-independent baseline $\bar{\beta}_S$, which does not introduce a bias \cite{Sutton}. A simple but effective baseline is the expectation of reward under the current policy, approximated by the mean of the sampled paths:
\begin{equation}\label{baseline}
    \bar{\beta}_S = \frac{1}{S} \sum_{s=1}^S\hat{J}(\pmb{\tau}^{s},\textbf{b}),
\end{equation}
which leads to:
\begin{equation}\label{Rew-Baseline}
	\nabla_\theta \mathbb{E}_{\pmb{\tau}} \left(\hat{J}\right) \approx \frac{1}{S} \sum_{s=1}^{S}  \left[ (\hat{J}(\pmb{\tau}^s,\textbf{b})-\bar{\beta}_S) \nabla_\theta \sum_{t=0}^{T-1} \text{log}\left( \pi_\theta(\textbf{x}_t^s,D^s_t)\right)\right].
\end{equation}
Using the above gradient of the expected reward with respect to the parameters the policy can now be iterative adjusted, in a steepest ascent framework or any of its variants (e.g. Adam):
\begin{equation}\label{update theta}
\theta_{k+1} := \theta_{k} +  \frac{\ell_k}{S} \sum_{s=1}^{S}  \left[ (\hat{J}(\pmb{\tau}^s, \textbf{b})-\bar{\beta}_s) \nabla_\theta \sum_{t=0}^{T-1} \text{log}\left( \pi_\theta(\textbf{x}_t^s,D^s_t)\right)\right],
\end{equation}
where $\ell_k$ is the adaptive learning rate. The algorithm that trains the policy network for a fixed backoff value  $\textbf{b}$ is outlined in Algorithm~\ref{alg:PG alg1}.
%
\begin{algorithm}[H]
\caption{Policy gradient for fixed backoff}\label{alg:PG alg1}
\begin{algorithmic}
\STATE {\bfseries Input:} Given optimization problem in Eq. (\ref{policyopt_con}) and modified objective in Eq. (\ref{eq:modified_obj}), initialize policy $\pi_\theta$ with parameters $\theta := \theta_0$, initial learning rate $\ell_0$, learning rate update rule $L(\ell_k)$, number of samples for the gradient approximation $S$, number of epochs $K$, tolerance $tol$, and backoffs $\textbf{b}$.\\
\FOR{$k= 1$ {\bfseries to} $K$} 
\STATE 1. Collect $\pmb{\tau}^s$ and $\hat{J}(\pmb{\tau}^s, \textbf{b})$ for samples $s={1,...,S}$.
\STATE 2. Update policy $\pi_\theta$: $\theta_{k+1} := \theta_{k} + $ 
\STATE $+\frac{\ell_k}{S} \sum_{s=1}^{S}  \left[ (\hat{J}(\pmb{\tau}^s, \textbf{b})-\bar{\beta}_s) \nabla_\theta \sum_{t=0}^{T-1} \text{log}\left( \pi_\theta(\textbf{x}_t^s,D^s_t)\right)\right]$.
\STATE 3. Update learning rate $\ell_{k+1}:=L(\ell_k)$.
\STATE 4. if $|\overline{\hat{J}}_{k+1}-\overline{\hat{J}}_{k}|\leq tol$ then $exit$, where $\overline{\hat{J}}$ refers to the sample average of $\hat{J}(\pmb{\tau}^s, \textbf{b})$.
\ENDFOR
\STATE {\bfseries Output:} Optimal policy $\pi_{\theta^*}$, with $\theta^* := \theta_{k+1}$.
\end{algorithmic}
\end{algorithm}

\subsection{Backoff iterations}

Solving Eq. (\ref{solve_backoff}) is complicated because there is no closed form solution, or even an explicit expression. To solve the root-finding problem in Eq. (\ref{solve_backoff}) efficiently, we  formulate it as a least-squares expensive black box optimization problem~\cite{Jones1998} and solve it via Bayesian optimization with the objective function
\begin{equation}\label{BO_obj}
\mathcal{F}(\pmb{\gamma}):=\left(\hat{F}_{lb}(\pmb{\gamma}) - (1-\alpha)\right)^2.
\end{equation}
Algorithm~\ref{alg:PG alg1} yields an optimal stochastic policy for fixed values of the backoffs. We propose Algorithm~\ref{alg:PG alg2} to iteratively adjust the backoffs to guarantee probabilistic constraint satisfaction.

%
%
%
A description of the steps conducted in Algorithm \ref{alg:PG alg2} is presented here:

\textbf{Step (1)}: The policy is trained with $\textbf{b} = 0$ using Algorithm~\ref{alg:PG alg1}.

\textbf{Step (2)}: The initial estimates of the backoff are computed by Monte-Carlo closed-loop simulations.

\textbf{Step (3)}: 

\textbf{(3i)} An initial set of backoffs parameter values $\pmb{\Gamma}=[\pmb{\gamma}^1, \ldots, \pmb{\gamma}^{N_\Gamma}]$ are used to construct a set of optimal policies $\pi_{\theta^{*}}(\cdot|\pmb{\gamma}^1)$\dots$\pi_{\theta^{*}}(\cdot|\pmb{\gamma}^{N_\Gamma})$, recall that $\pmb{\gamma}$ parametrizes $b_{j,t} := \gamma_{j}~b^0_{j,t}~\forall j, t$.

\textbf{(3ii)} For each optimal policy $\pi_{\theta^{*}}(\cdot|\pmb{\gamma}^i), i \in \{1,...,N_\Gamma\}$, the lower bound on the ECDF, $\hat{{F}}_{lb}(\pmb{\gamma}^i),i \in \{1,...,N_\Gamma\}$, is computed by Monte Carlo (a cross-validation of sorts), along with the squared residual in Eq.~\ref{BO_obj}. This yields a set of backoff parameter values $\pmb{\Gamma}=[\pmb{\gamma}^1, \ldots, \pmb{\gamma}^{N_\Gamma}]$ that correspond to residuals ${\textbf{F}}=[\mathcal{F}(\pmb{\gamma}^1), \ldots,\mathcal{F}(\pmb{\gamma}^{N_\Gamma})]$ for Eq. (\ref{BO_obj}). This is the initial sample set used to map backoff values to residual values via a Gaussian process regression, and subsequently solved via a Bayesian optimization framework to enforce Eq. (\ref{solve_backoff}).

\textbf{Step (4)}: For each $m^{th}$ iteration:

\textbf{(4i)} A Gaussian process which maps the backoff values to the squared residuals in Eq.~(\ref{BO_obj}) is constructed. We compute the objective function value as in Eq. (\ref{BO_obj}) to find $\mathcal{F}(\pmb{\gamma}^{N_\Gamma+m-1}):=\left(\hat{F}_{lb}(\pmb{\gamma}^{N_\Gamma+m-1}) - (1-\alpha)\right)^2$.

\textbf{(4ii)} 
We conduct a Bayesian optimization step, via lower confidence bound minimization~\cite{Frazier2018} of the GP. This allows us to obtain the next value for $\pmb{\gamma}^{N_\Gamma+m}$. 

\textbf{(4iii)} We compute a new policy $\pi(\cdot|\pmb{\gamma}^{N_\Gamma+m})$ using Algorithm \ref{alg:PG alg1}.

\textbf{(4iv)} We compute the new backoff $\textbf{b}$ with the new value of the backoff parameter $\pmb{\gamma}^{N_\Gamma+m}$. We also compute the new value for the residuals (objective) $\mathcal{F}$. Data matrices $\pmb{\Gamma}$ and $\textbf{F}$ are updated. The algorithm goes back to \textbf{(i)} where a new GP is constructed incorporating the new values of $\pmb{\gamma}^{N_\Gamma+m}$ and $\mathcal{F}$ and the algorithm proceeds until some tolerance is achieved (e.g. $\mathcal{F}(\pmb{\gamma}^{N_\Gamma+m}) \leq tol$).

It should be noted that every time a new backoff is computed the policy is re-optimized. This may look inefficient at first glance, however the convergence is achieved fast, as at every iteration, the previous policy is used as an initial guess for the next iteration, making the first iteration the most expensive one. Additionally, because the problem is treated as an expensive black-box optimization problem the number of iterations from this outer loop is in general quite small (in our examples no more than 13 iterations). (see section \ref{subse:policy innit}). By the end of Algorithm \ref{alg:PG alg2}, a probabilistically constrained policy will have been constructed.
\begin{algorithm}
\caption{Backoff-Based Policy Optimization}\label{alg:PG alg2}
\begin{algorithmic}
\STATE{\bfseries Input:} Initialize policy parameter $\theta := \theta_{0}$, initial learning rate $\ell_0$, learning rate update rule $L(\ell_k)$, define initial data matrix $\pmb{\Gamma}:=[\pmb{\gamma}^1, \ldots, \pmb{\gamma}^{N_\Gamma}]$ with $N_\Gamma$ values for $\pmb{\gamma}$, $0<\delta<1$, $0<\alpha< 1$, tolerance $tol_0$,  maximum number of backoff iterations $M$ and $S$ number of Monte Carlo samples to compute $\hat{F}_{lb}$, and number of epochs $K$.\\
    \STATE \textbf{1.} { Perform policy optimization} with $\textbf{b} = \textbf{0}$ using Algorithm~\ref{alg:PG alg1}, to obtain nominal policy $\overline{\pi}_{{\theta}^*}$. 
     \STATE \textbf{2.} Estimate initial backoffs using $S$ samples generated by Monte-Carlo simulations from the state trajectories of the nominal policy:\\ $b^0_{j,t} := F_{\mathbf{g}_{j,t}}^{-1}(1-\delta) - \bar{g}_{j,t}(\textbf{x}_t),~\forall j,t$ 
     \STATE \textbf{3.} {
     (i) Use the initially proposed set of backoff parameter values $\pmb{\Gamma}=[\pmb{\gamma}^1, \ldots, \pmb{\gamma}^{N_\Gamma}]$ ({\it Notice} that the backoffs  are $b_{j,t} := \gamma_{j}~b^0_{j,t}~\forall j, t$)  to obtain optimal policies $\pi_{\theta^{*}}(\cdot|\pmb{\gamma}^1)$\dots$\pi_{\theta^{*}}(\cdot|\pmb{\gamma}^{N_\Gamma})$
     
     (ii) Evaluate the policies $\pi_{\theta^{*}}(\cdot|\pmb{\gamma}^1)$\dots$\pi_{\theta^{*}}(\cdot|\pmb{\gamma}^{N_\Gamma})$ and compute their corresponding lower bounds $[\hat{F}_{lb}(\pmb{\gamma}^1), \ldots,\hat{F}_{lb}(\pmb{\gamma}^{N_\Gamma})]$ by Monte-Carlo using Eq.~\ref{solve_backoff}, and compute their residuals ${\textbf{F}}=[\mathcal{F}(\pmb{\gamma}^1), \ldots,\mathcal{F}(\pmb{\gamma}^{N_\Gamma})]$ using Eq.~\ref{BO_obj}}.  
      \STATE \textbf{4.} Perform Bayesian Optimization: \FOR{$m= 1$ {\bfseries to} $\dots$} 
        \STATE i) Construct a mapping from the backoffs parameter values $\pmb{\Gamma}=[\pmb{\gamma}^1, \ldots, \pmb{\gamma}^{N_\Gamma+m-1}]$ to their residuals         ${\textbf{F}}=[\mathcal{F}(\pmb{\gamma}^1), \ldots,\mathcal{F}(\pmb{\gamma}^{N_\Gamma + m-1})]$
        by using a GP regression.
        \STATE  ii) Perform Bayesian optimization over the GP to minimize $\mathcal{F}(\pmb{\gamma}^{N_{\Gamma}+m}):=\left(\hat{F}_{lb}(\pmb{\gamma}^{N_{\Gamma}+m}) - (1-\alpha)\right)^2$.
        \STATE iii)  { Perform policy optimization} with $b_{j,t} := \gamma_{j}^{N_\Gamma+m}~b^0_{j,t}~\forall j, t$ using Algorithm~\ref{alg:PG alg1}, to obtain  policy ${\pi}_{{\theta}^*}(\cdot|\pmb{\gamma}^{N_\Gamma+m})$. 
        \STATE  iv) Update data matrices $\pmb{\Gamma}:=[\pmb{\Gamma},\pmb{\gamma}^{N_{\Gamma}+m}]$ and ${\textbf{F}}:=[\textbf{F}, \mathcal{F}(\pmb{\gamma}^{N_{\Gamma}+m})]$ by Monte-Carlo using Eq.~\ref{solve_backoff}.  
    \IF {$\mathcal{F}(\pmb{\gamma}^{N_{\Gamma}+m}) \leq tol_0$} 
    \STATE {\it exit}
    \ENDIF
    \ENDFOR
\STATE{\bfseries Output:} policy $\pi^*_\theta:=\pi_{\theta^*}^{N_{\Gamma}+m}$.
%
\end{algorithmic}
\end{algorithm}
%


\subsection{Policy initialization}\label{subse:policy innit}
Reinforcement learning methods (particularly policy gradient) are computationally expensive; mainly because initially the agent (or controller in our case) explores the control action space randomly. In the case of process optimization and control, it is possible to use a preliminary controller, along with supervised learning to hot-start the policy, and significantly speed-up convergence. The initial parameterization for the policy (before {\textbf{ Step (1)}}) is trained in a supervised learning fashion where the states are the inputs and the control actions are the outputs. This has been further discussed in \cite{petsagkourakis2020constrained, Paulson_DNN, LUCIA_dnn}
%
 %
 
\section{Case Studies}
The case studies in this paper focuses on the photo-production of phycocyanin synthesized by cyanobacterium \textit{Arthrospira platensis}. The two case studies are separated regarding the type of uncertainty: The case study 1 considers parametric uncertainty and model is considered to be known. The second case study considers that only data is available, where additive disturbance and measurement noise are present (no knowledge of the system's equations). Additionally, in the two case studies different penalizations of the constraints are implemented: Case study 1 uses Eq.(\ref{eq:modified_obj}) with arbitrary $\kappa=0.1$ and $p=1$, and Case study 2 uses Eq.(\ref{eq:modified_obj}) with arbitrary $\kappa =1.$ and $p=2$.

Phycocyanin is a high-value bioproduct and its biological function is to enhance the photosynthetic
efficiency of cyanobacteria and red algae. It has applications as a natural colorant to replace other toxic synthetic pigments in both food and cosmetic production. Additionally, the pharmaceutical industry considers it beneficial because of its unique antioxidant, neuroprotective, and anti-inflammatory properties.

The  dynamic system consists of the following system of ODEs describing the evolution of the concentration ($c$) of biomass ($x$), nitrate ($N$), and product ($q$). The dynamic model is based on Monod kinetics, which describes microorganism growth in nutrient sufficient cultures, where intracellular nutrient concentration is kept constant because of the rapid replenishment. We assume a fixed volume fed-batch. The manipulated variables as in the previous examples are the light intensity ($u_1=I$) and inflow rate ($u_2=F_N$). The mass balance equations are
\begin{align}
     &\dfrac{dc_x}{dt} =  u_m \dfrac{I}{I + k_s + I^2/k_i} \dfrac{c_x c_N}{c_N + K_N} - u_dc_X\label{biomas}\\
     &\dfrac{dc_N}{dt} = -Y_{N/X}  \dfrac{u_m I}{I + k_s + I^2/k_i} \dfrac{c_x c_N}{c_N + K_N} + F_N \label{Nitrogen}\\
     &\dfrac{dc_q}{dt} =
       \dfrac{k_m~I}{I + k_{sq} + I^2/k_{iq}} {c_x} - \dfrac{k_d c_{q}}{C_N + K_{N_q}}       \label{product}
\end{align}
The parameter values are adopted from~\cite{Bradford2019}. 


\subsection{Case Study 1}
Uncertainty is assumed for the initial concentration, where
$\begin{bmatrix}
c_x(0) & c_N(0)
\end{bmatrix} \sim \mathcal{N}(    \begin{bmatrix}
1. & 150.
\end{bmatrix}$, $\text{diag}(1\times 10^{-3}, 22.5))$ and $c_q(0) = 0$.
Additionally, 10\% of parametric uncertainty for  the system is  assumed:
$        \dfrac{k_s}{(\mu mol/m^{2}/s)} \sim \mathcal{N}(178.9, 17.89),
        \dfrac{k_i}{(mg/L)} \sim \mathcal{N}(447.1, 44.71),
        \dfrac{k_N}{(\mu mol/m^{2}/s)} \sim \mathcal{N}( 393.1, 39.31)$. This type of uncertainty is common in engineering settings, as the parameters are obtained using experimental data, and they are subject to their respective confidence regions after they are estimated using regression techniques.
The objective function (reward) in this work is to maximize the product's concentration ($c_q$) at the end of the batch. The objective is additionally penalized by the change of the control actions $\textbf{u}(t) = \left[ I, F_N\right]^{T}$. As a result the reward is:
 \begin{equation}\label{reward_bio}
 \begin{split}
     &R_t = -||\Delta \textbf{u}_t||_r, R_T = c_q(T)\\& r = diag(3.125\times10^{-8}, 3.125 \times 10^{-6})\\ 
 \end{split}
 \end{equation}
 where $t\in \{0, T-1\}$ and $\Delta \textbf{u}_t = \textbf{u}_t - \textbf{u}_{t-1}$. 
 The constraints in this work for each time step are $c_N \leq 800$ and $c_q \leq 0.011c_X$. This constraints have been normalized as:
 \begin{equation}
         g_{1,t} =\dfrac{c_N}{800} - 1\leq 0, ~~~
         g_{2,t} = \dfrac{c_q}{0.011c_X} - 1\leq 0
 \end{equation}
and the joint chance constraint is meant to be satisfied with probability $99\%$ ($\alpha = 0.01$) and confidence level is $99\%$ ($\epsilon = 0.01$). The constraints are added as a penalty with $\kappa = 1$ using (\ref{eq:modified_obj}) and $p=1$. The control actions are constrained to be in the interval $0\leq F_N \leq 40$ and $120 \leq I\leq 400$, these constraints are considered to be hard. The control policy RNN is designed to contain 4 hidden layers, each of which comprises 20 neurons with a leaky rectified linear unit (ReLU) as activation function. A unified policy network with diagonal variance is utilized such that the control actions share memory and the previous states are used from the RNN (together the current measured states). The  computational cost for each control action online is insignificant since it only requires the evaluation of the corresponding RNN. First the algorithm computes the policy for the backoffs to be zero ($\textbf{b} =0$), then the backoffs are updated according to the Algorithm~\ref{alg:PG alg2}. The parameters for the trainings are: $M = 200$, $S = 1000$, $K = 200$, $tol=tol_0=10^{-4}$, $N_\gamma = 5$ and the two previous states and controls are used from the policy. Additionally, the Gaussian process has zero mean as prior, squared-exponential (SE) kernel as the covariance function, and the inputs-outputs are normalized using its mean and variance respectively. After the completion of the training the backoffs have been computed to satisfy (\ref{solve_backoff}). In a rather small number of iterations the backoff values managed to force the $F_{lb}$ to 0.99.

Now, the actual closed-loop constraint satisfaction can be depicted in Fig.~\ref{fig:h1}(a) and Fig.~\ref{fig:h2}(a), where the shaded areas are the 98\% and 2\% percentiles. Notice, that even though the figures for both methods look similar the constraint satisfaction is significantly different, this can be observed in Fig.~\ref{fig:h1}(b) and Fig.~\ref{fig:h2}(b), where the region that the violation of constraints occurs has been zoomed in. This result is also quantitatively depicted in Table~\ref{tab:clsatisfaction}. The column labelled `actual' is the probability of constraint violation that corresponds to the fraction of the 1000 Monte-Carlo trajectories that satisfied both of the constraints (see Eq. (\ref{ECDF})). The column labelled `desired' is the goal for the probability for constraint satisfaction. It is clear that when backoffs are not applied, almost 50\% of the constraints are violated. On the other hand when backoffs are applied then all the constraints are satisfied, which is an expected result as the goal was the lower bound of ECDF ($F_{lb}$) to be 0.99, which means that the actual probability is equal or higher. 
\begin{figure}[H]
\includegraphics[width=\linewidth]{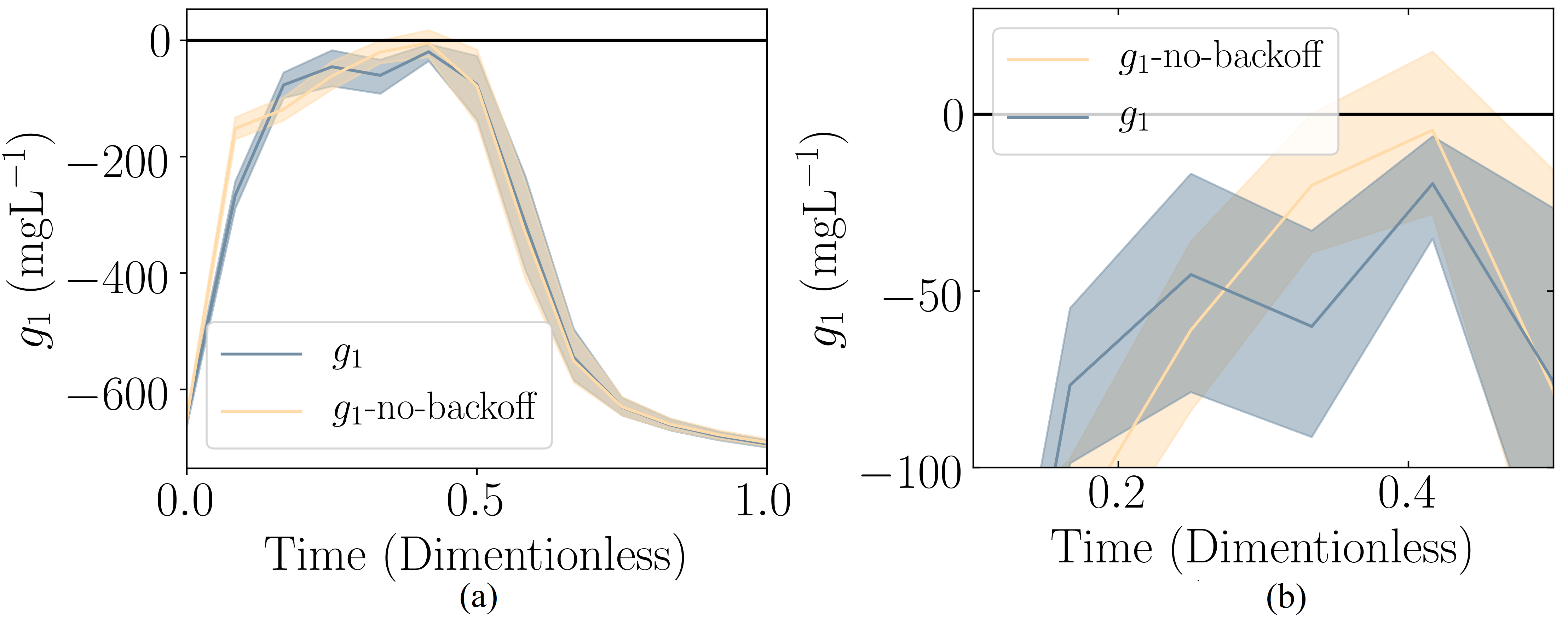}
    \caption{Case study 1: Constraints when backoffs are applied (blue) and when there are set to be zero (yellow) for $g_{1,t}$ (a) and  zoomed in the region that the violation of constraints occur (b). The shaded areas are the 98\% and 2\% percentiles.}
    \label{fig:h1}
\end{figure}
\begin{figure}[H]
\includegraphics[width=\linewidth]{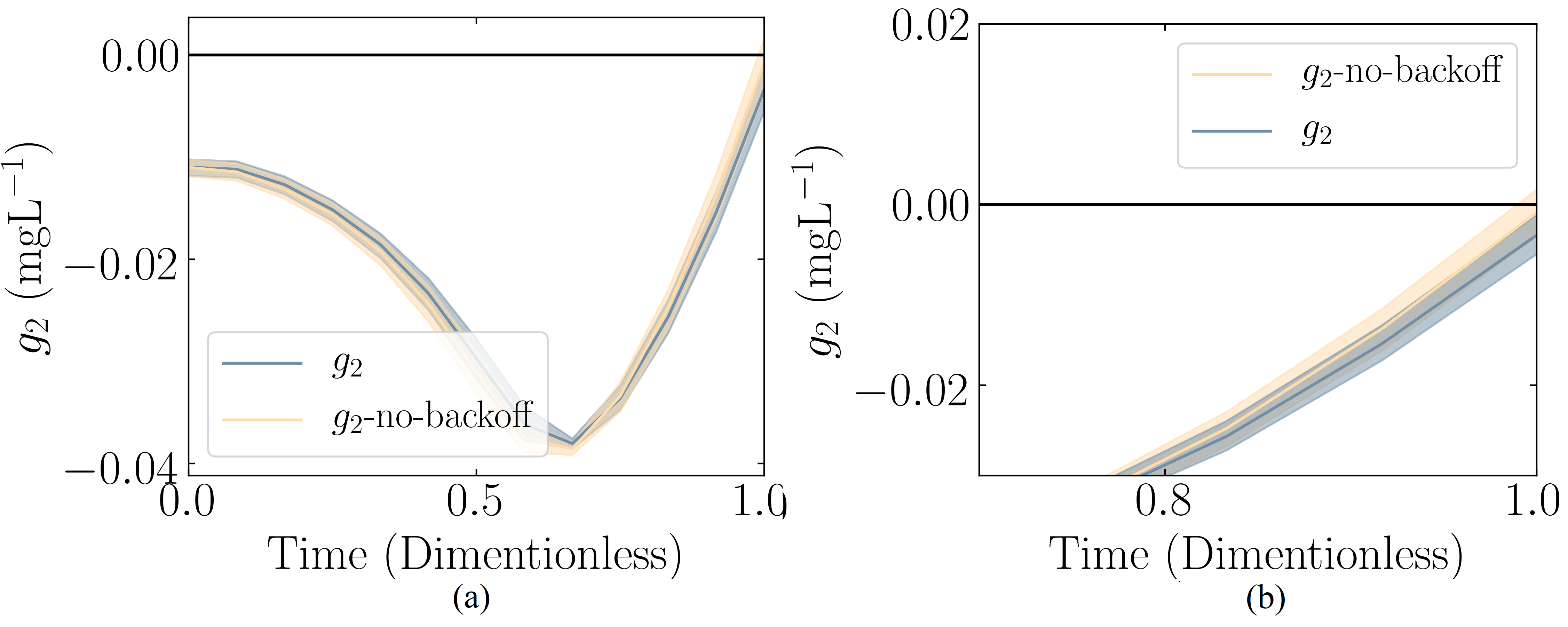}
    \caption{Case study 1: Constraints when backoffs are applied (blue) and when there are set to be zero (yellow) for $g_{2,t}$ (a), zoomed region where the violation of constraints occurs (b). The shaded areas are the 98\% and 2\% percentiles.}
    \label{fig:h2}
\end{figure}
The backoff values for each update are shown in Fig. \ref{fig:backoff1} (a, b), where the red-dashed represents the converged final value. It should be noted that the final value for the product $c_q$~(\ref{reward_bio}) is 0.163 and 0.167 when backoffs are applied and when they are not. This difference is to be expected given that in the nominal case ($\textbf{b} = 0$), the policy allows the violation of constraints which lead to an increase in the concentration of the product. 

\begin{figure}[H]
\centering
\includegraphics[width=\linewidth]{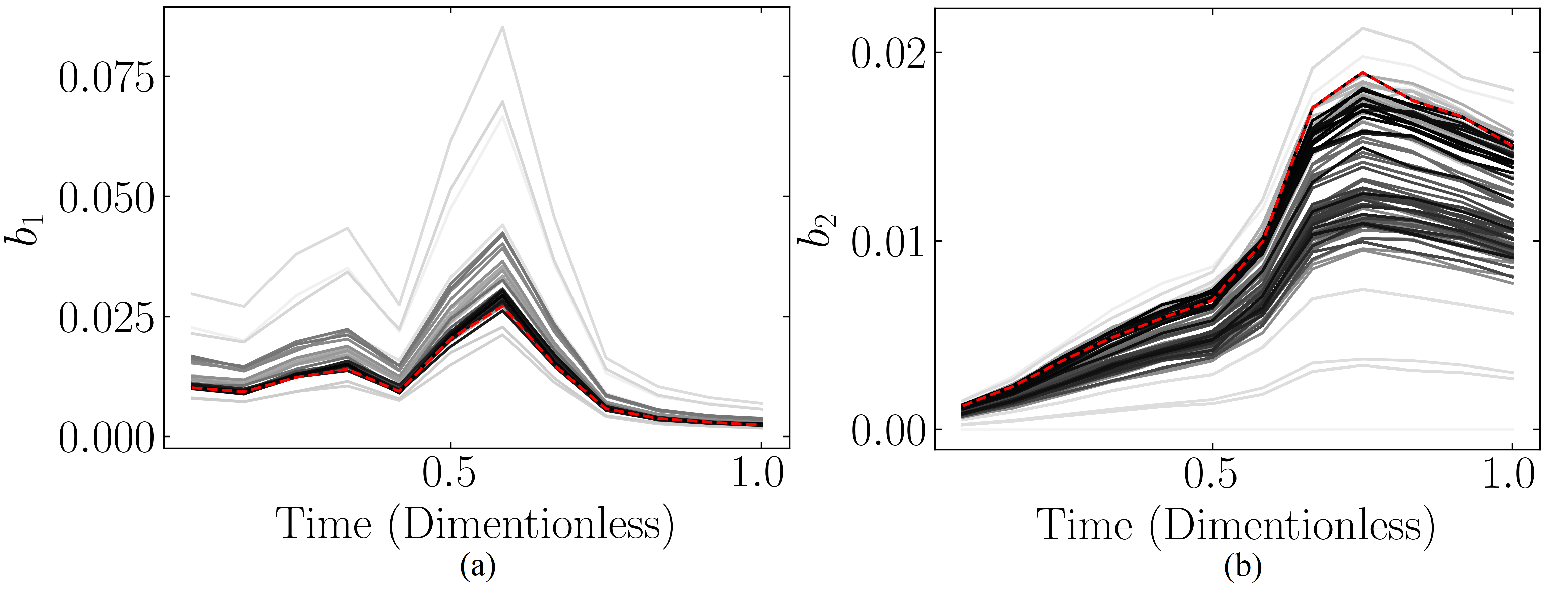}
    \caption{Case study 2: The lines are plotted over number of epochs for $b_{1,t}$ (a) and ,$b_{2,t}$ (b) which are faded out towards earlier iterations.}
    \label{fig:backoff1}
\end{figure}

\subsection{Case Study 2}
The second case study considers the same system but this time no equation is assumed to be available. In this case, a Gaussian process\cite{Rasmussen2006} is used to model the available data. Different kernels can have different effects in the performance, in this case the Mat\'{e}rn32\cite{Rasmussen2006} kernel is employed. The data was generated using the system in Eqs. (\ref{biomas}-\ref{product}), and an additive normally distributed disturbance ($\textbf{w}$) and measurement noise ($\textbf{v}$) with zero mean and $\pmb{\Sigma}_\textbf{w} = \text{diag}(4 \times 10^{-4}, 0.1, 1 \times 10^{-8})$, $\pmb{\Sigma}_\textbf{v} = \text{diag}(4 \times 10^{-5}, 0.01, 1 \times 10^{-9})$. The same uncertainty for the initial conditions is assumed. 

Bayesian frameworks have been used in reinforcement learning~\cite{Chua, Janner} as they can model both the epistemic and aleatoric uncertainty (in Gaussian processes, usually homoscedastic aleatoric uncertainty is considered). One of the fundamental steps here is the propagation of the uncertainty for each step. The nature of Gaussian processes has been exploited in \cite{Umlauft2018, Hewing2019, Bradford2019}, where the generated samples from the GP are conditioned in the posterior (without noise). Then, the GP predicts a mean $\pmb{\mu}_\textbf{x}$  and variance $\pmb{\Sigma}_\textbf{x}$ and the state $\textbf{x}$ is sampled from a normal distribution:  $\textbf{x}\sim\mathcal{N}(\pmb{\mu}_\textbf{x},\pmb{\Sigma}_\textbf{x})$. In this setting 8 episodes were randomly generated using a Sobol sequences~\cite{SOBOL2001271} for the control variables of all 12 time intervals and the initial conditions to train the Gaussian process.

After the training the actual closed-loop constraint satisfaction can be validated. Fig.~\ref{fig:h12}(a) and Fig.~\ref{fig:h22}(a) illustrate the path constraints for each time interval, where the shaded areas are the 98\% and 2\% percentiles. The results are compared with the case in absence of backoffs ($\textbf{b}=0$). It should be noticed that in the absence of backoffs the mean of the constraints' samples barely satisfies the bounds. The difference is clear and more apparent in Fig.~\ref{fig:h22}(a). This can further be noticed in Fig.~\ref{fig:h12}(b) and Fig.~\ref{fig:h22}(b), where focus has been put on the region of constraint violation. Table~\ref{tab:clsatisfaction} shows the probability of constraint violation that corresponds to the fraction of the 1000 Monte-Carlo trajectories that satisfied the constraints (see Eq. (\ref{ECDF})) compared to the `desired' probability for constraint satisfaction. The absence of constraint tightening results in only 24\% of constraint satisfaction compare to our proposed method where 97\% are satisfied. Notice that the desired probability designed in terms of the lower bound of the ECDF ($F_{lb}$) is set to be 0.95.

\begin{figure}[H]
\includegraphics[width=\linewidth]{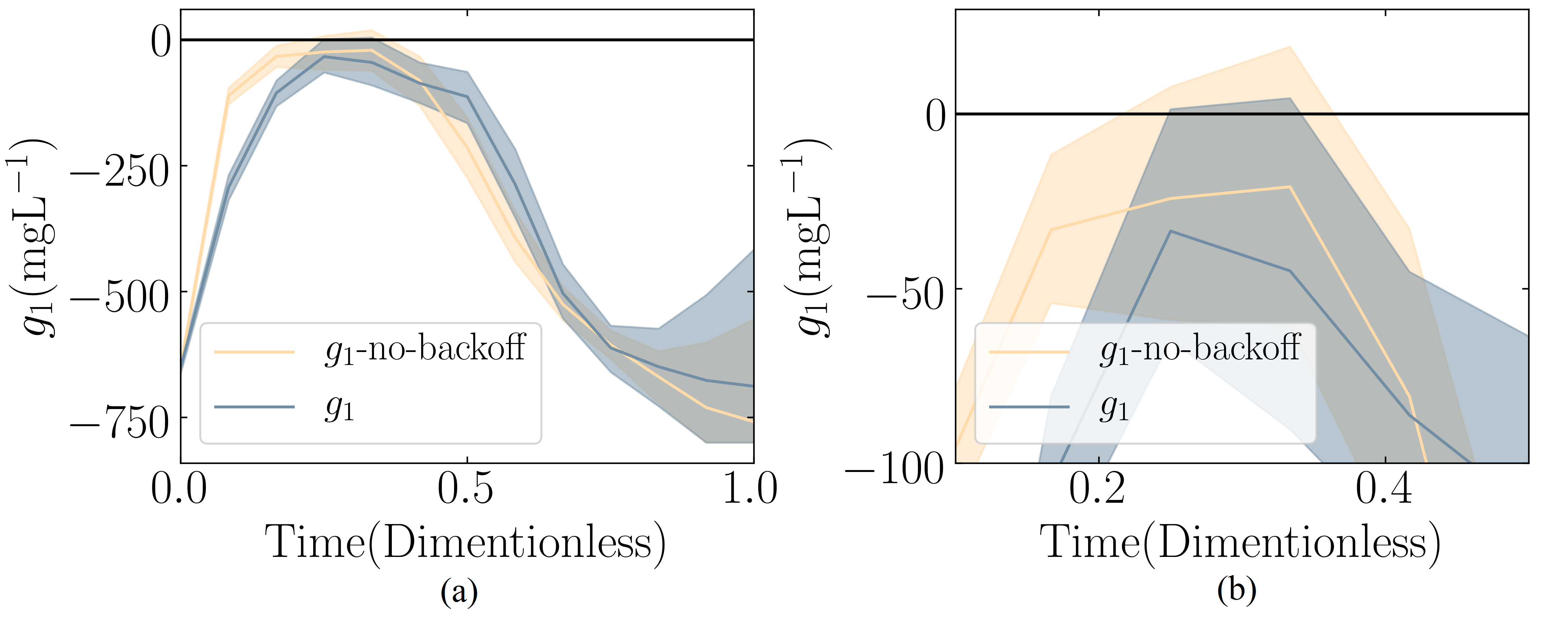}
    \caption{Case study 2: Constraints when backoffs are applied (blue) and when they are set to be zero (yellow) for $g_{1,t}$ (a), zoomed region where the violation of constraints occurs (b). The shaded areas are the 98\% and 2\% percentiles.}
    \label{fig:h12}
\end{figure}
\begin{figure}[H]
\includegraphics[width=\linewidth]{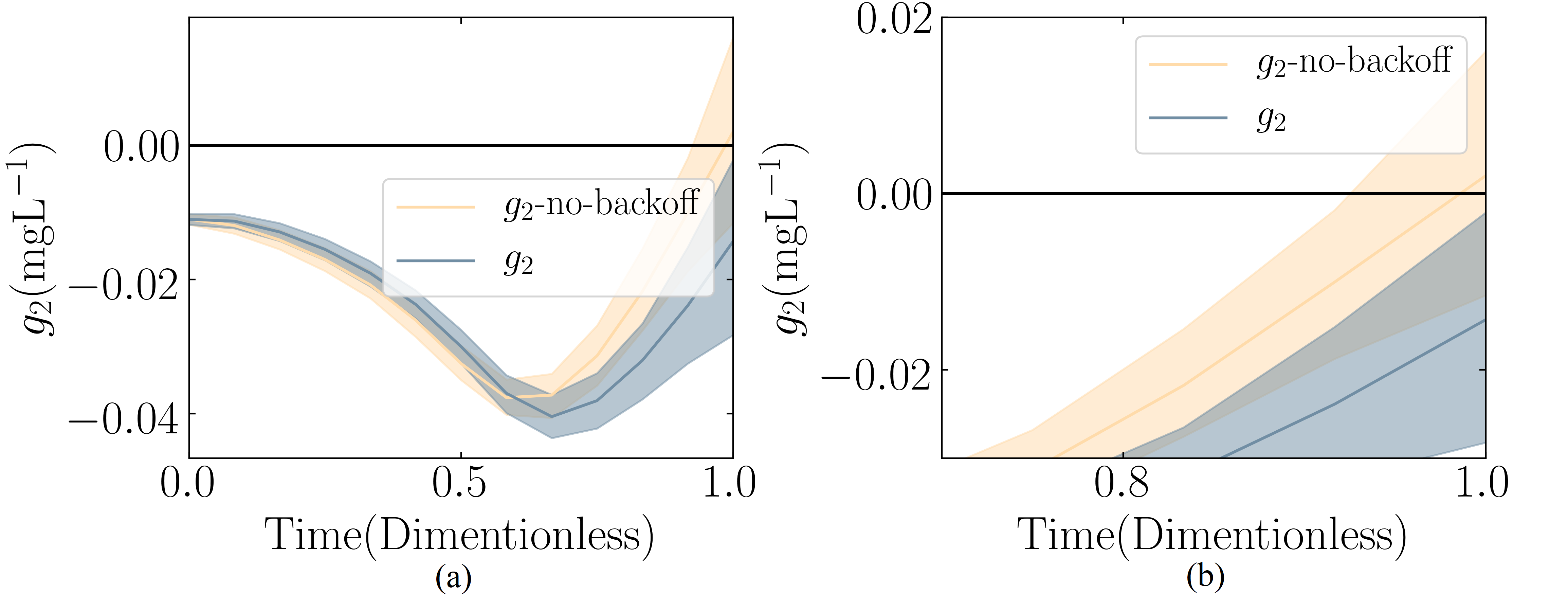}
    \caption{Case study 2: Constraints when backoffs are applied (blue) and when they are set to zero (yellow) for $g_{2,t}$ (a), zoomed region where the violation of constraints occurs (b). The shaded areas are the 98\% and 2\% percentiles.}
    \label{fig:h22}
\end{figure}
The backoff values for each update are shown in Fig. \ref{fig:backoff2} (a, b), where the red-dashed represents the converged final value. It should be noted that the final value for the product $c_q$~(\ref{reward_bio}) is 0.153 and 0.171 when backoffs are applied and when they are not. 
\begin{figure}[H]
\centering
\includegraphics[width=\linewidth]{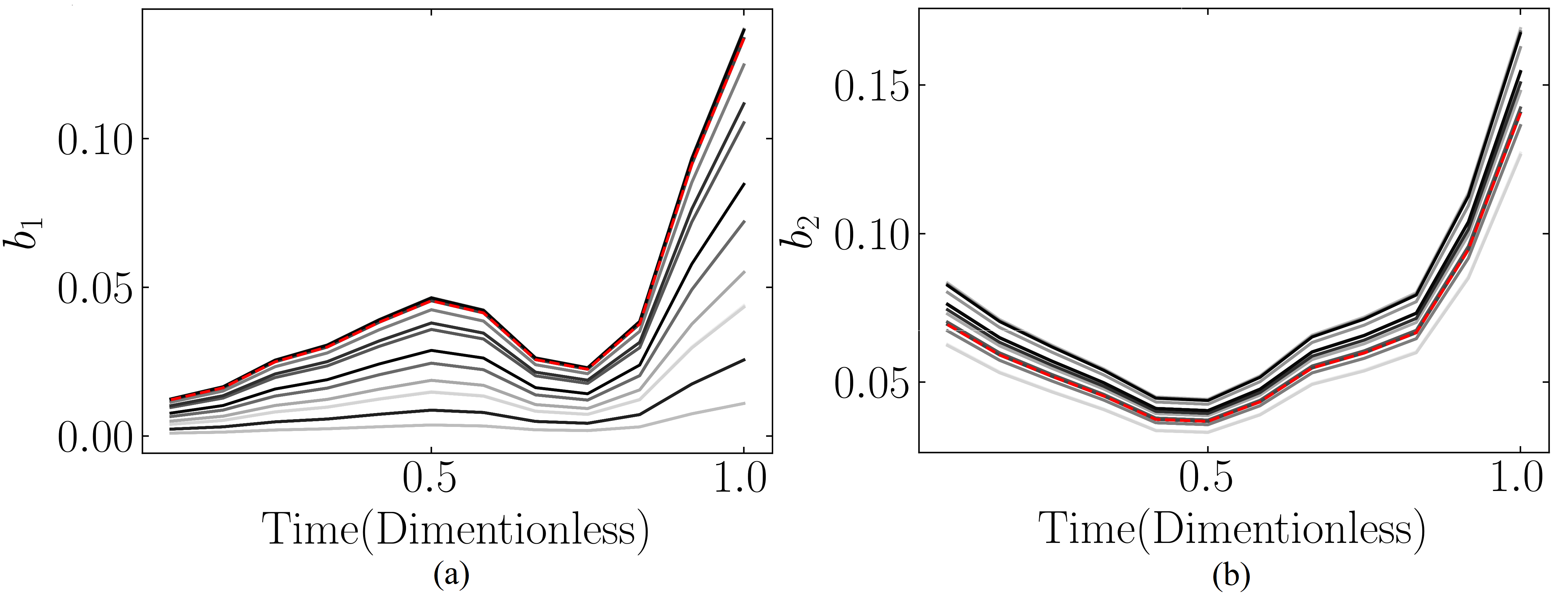}
    \caption{Case study 2: The lines are plotted over number of epochs for $b_{1,t}$ (a) and ,$b_{2,t}$ (b) which are faded out towards earlier iterations.}
    \label{fig:backoff2}
\end{figure}

\begin{table}[h]
    \centering
    \begin{tabular}{lll}
          \hline
      Case Study & Desired & Actual \\
      \hline
      Case Study 1: Parametric Nominal   & 0.99 & 0.51$^*$\\
      Case Study 1: Parametric Proposed   & 0.99 & 1.00\\
     Case Study 2: GP Nominal   & 0.95 & 0.24$^*$\\
     Case Study 2: GP Proposed   & 0.95 & 0.97\\
            \hline
    \end{tabular}
    \caption{Comparison of closed-loop constraint satisfaction. The symbol $^*$ corresponds to the cases that the constraint satisfaction is less than the desired.}
    \label{tab:clsatisfaction}
\end{table}

The algorithm is implemented in Pytorch \cite{NEURIPS2019_9015} version 0.4.1.  Adam~\cite{Kingma2014} is employed to compute the network's parameter values using a step size of $10^{-2}$ with the rest of hyperparameters at their default values. 

\section{Conclusions}
In this paper we address the problem of finding a policy that can satisfy constraints with high probability. The proposed algorithm - chance constrained policy optimization (CCPO) - uses constraint tightening by applying backoffs to the original feasible set. Backoffs restrict the perceived feasible space by the controller, and allow guarantees on the satisfaction of chance constraints. We find the smallest backoffs (least conservative) that still guarantee the desired probability of satisfaction by stating the root-finding problem as a black-box optimization problem. This allows the algorithm to construct a policy that can guarantee the satisfaction  of  joint  chance  constraints  with  a user-defined probability of at least $1-\alpha$ and a user-defined confidence of at least $1-\epsilon$. Furthermore, the proposed methodology can be combined with other penalty or Lagrangian approaches for constrained policy search.
Being able to solve constraint policy optimization problems with high probability satisfaction has been one of the main bottlenecks for the wider use of reinforcement learning in engineering applications. This work aims to take a step towards applying RL to the real world, where constraints on policies are necessary for the sake of safety and product quality.


\section*{Acknowledgment}
This project has received funding from the EPSRC projects (EP/R032807/1) and (EP/P016650/1).

\bibliographystyle{IEEEtran}
\bibliography{IEEEabrv, ifacconf}

\clearpage

\appendices
\section{Bayesian optimization}\label{BO}

Bayesian optimization is a popular approach for black-box optimization \citep{Jones1998}. A review of Bayesian optimization can be found in \cite{Shahriari}. 

In this paper Bayesian optimization is exploited to obtain the minimizer $\pmb{\gamma}^*$ of $\mathcal{F}(\pmb{\gamma}) = \left(\hat{F}_{lb} - (1-\alpha)\right)^2$:
\begin{equation}
\pmb{\gamma}^* \in \argmin_{\pmb{\gamma}} \mathcal{F}(\pmb{\gamma})
\end{equation}

The function $\mathcal{F}(\pmb{\gamma})$ to be minimized can only be observed through an unbiased noisy observation based on MC estimates of $\hat{F}_{lb}$:
\begin{equation}
y = \mathcal{F}(\mathbf{x}) + \omega
\end{equation}
where $\omega$ is assumed to be Gaussian white noise. The noise $\omega$ is assumed to be unknown and estimated within the GP framework.   
Commonly Gaussian processes (GP) are employed as nonparametric models. To start the algorithm we first require some data to build the initial GP, i.e. the function $\mathcal{F}(\cdot)$ is queried at $N_{\Gamma}$ points. In general the input data-points $\pmb{\Gamma}=[\pmb{\gamma}_1,\dots,\pmb{\gamma}_{N_{\Gamma}}]$ are selected based on a space-filling design, such as a Latin hybercube design \citep{hypercube}. From this we obtain the corresponding responses $\hat{\mathcal{F}}_{lb}=[\mathcal{F}(\pmb{\gamma}^1), \ldots,\mathcal{F}(\pmb{\gamma}^{N_{\Gamma}})]$. A GP model can then be trained from the input-output data. In particular the hyperparameters of the GP were trained using maximum likelihood estimation. The GP model can then be utilized to obtain the Gaussian distribution of $\mathcal{F}(\pmb{\gamma})$ at an arbitrary query point $\pmb{\gamma}$ \citep{Rasmussen2006}:
\begin{equation}
\mathcal{F}(\pmb{\gamma})|\pmb{\Gamma},\hat{\mathcal{F}}_{lb} \sim \mathcal{N}\left(\mu_{GP}(\pmb{\gamma}),\sigma^2_{GP}(\pmb{\gamma})\right) 
\end{equation}
where $\mu_{GP}(\pmb{\gamma})$ and $\sigma^2_{GP}(\pmb{\gamma})$ are the mean and variance prediction of the GP respectively. 
In this setting the approach sequentially selects a location $\pmb{\gamma}$ at which to query $\mathcal{F}(\cdot)$ and observe $y$. After a selected number of iterations $M$ the algorithm returns a best-estimate of $\pmb{\gamma}^*$. To accomplish this the GP of $\mathcal{F}(\cdot)$ is iteratively updated from the available data of $\mathcal{F}(\cdot)$. One could simply sample at the minimum of the mean function $\mu_{GP}(\pmb{\gamma})$. However, sampling  at a point with higher uncertainty could yield a lower minimum, i.e. there is a trade-off between sampling at points with low values of the mean function $\mu_{GP}(\pmb{\gamma})$ and high values of the variance function $\sigma^2_{GP}(\pmb{\gamma})$. The selection of the query points is given by so-called acquisition function, for which we used the lower confidence bound:
\begin{equation}
\pmb{\gamma}_m = \argmin_{\pmb{\gamma}} \mu_{GP}(\pmb{\gamma}) - 3 \sigma_{GP}(\pmb{\gamma})
\end{equation}
where $\pmb{\gamma}_m$ denotes the query point at iteration $m$. 
Note that the query point at each iteration are chosen at points that are predicted to be low by the mean function, but could also potentially yield lower values according to the variance function.

\newpage
\section{Nomenclature}

\begin{table}[htbp]
  \centering
  \caption{Nomenclature}
    \begin{tabular}{ll}
        \hline
    \multicolumn{1}{l}{Symbol} & \multicolumn{1}{l}{Description} \\
        \hline
    $N_x$  & Size of variable $\textbf{x}_t$ \\
    $\textbf{x}_t$  & state  at time $t$ \\
    $\textbf{u}_t$  & manipulated variables at time $t$ \\
    $\textbf{w}_t$  & External Disturbances at time $t$ \\
    $p(a|b)$  & probablity of event $a$ given $b$ \\
     $\mathcal{N}(\cdot,\cdot)$ & Normal distribution defined by mean and covariance \\
    $\sim$ & distributed according to; example: $x \sim \mathcal{N}(\mu, \sigma^2) $\\
    $\mathcal{D}$     & Data available \\
    $\mathbb{E}$     & Expectation \\
    $\mathbb{P}$      & Probability \\
    $\mathbb{X}$      & Feasible space for states variables \\
    $\mathbb{U}$      & Feasible space for manipulated variables\\
    $\cap$ & Intersection of sets \\
    $\textgamma$ & discount factor \\
    $R_t$ & Reward at time $t$ \\
    $g_j,t$ & Constraint $j$ at time $t$ \\
    $\pi_{\theta}(\cdot)$ & Policy parametrized by $\theta$ \\
    $b_j,t$ & Backoff of constraint $j$ at time $t$ \\
    $\pmb{\tau}$ & Joint random variable of states, controls and reward \\
    $1-\alpha$  & The probability of constraint satisfaction \\
    $1-epsilon$ & Confidence level \\
    $F(c)$    & Cumulative distribution function (CDF)  at $c$ \\
    $\mathds{1}(\cdot)$     & Indicator function \\
    Bin   & Binomal distribution \\
    $F_S$  & Cumulative distribution function (ECDF) for $F$ with $S$ samples \\
    $\hat{F}_S$ & Realization of $F_S$  \\
    betainv & Inverse comulative distribution of Beta probability distribution \\
    $F_{lb}$ & Lower bound of ECDF $F_S$ \\
    \textbackslash{}hat{F}\_{lb} & Realization of F\_lb \\
    pmb{\textbackslash{}gamma } & Parametrization of backoffs \\
    $||\cdot||_p$ & $l_p$ norm \\
    $\nabla_{\theta}$ & Parial derivatives with respect to $\theta$ \\
        \hline
    \end{tabular}%

  \label{tab:addlabel}%
\end{table}%

\end{document}